\documentclass[runningheads]{llncs}
\usepackage{graphicx}
\usepackage{subfigure}
\usepackage{amsmath}
\usepackage{amsfonts}
\usepackage{forest}
\usepackage{algorithmicx}
\usepackage{algpseudocode}
\usepackage{tikz}
\usetikzlibrary{positioning,trees,fit}

\algtext*{EndFunction}
\algtext*{EndProcedure}
\algtext*{EndFor}
\algtext*{EndWhile}
\algtext*{EndIf}
\algloopdefx{NIf}[1]{\textbf{if} #1 \textbf{then}}
\algloopdefx{NElse}{\textbf{else}}
\algloopdefx{NElseIf}{\textbf{else if}}
\algloopdefx{NForAll}[1]{\textbf{for each} #1 \textbf{do}}
\algloopdefx{NWhile}[1]{\textbf{while} #1 \textbf{do}}
\algloopdefx{NFor}[1]{\textbf{for} #1 \textbf{do}}

\newcommand\Oh{\mathcal{O}}
\newcommand\pop{\mathsf{pop}}

\begin{document}
\title{Optimal Skeleton Huffman Trees Revisited\thanks{Supported by the Russian Science Foundation (RSF), project 18-71-00002.}}
\author{Dmitry Kosolobov\inst{1} \and Oleg Merkurev\inst{1}}
\authorrunning{D. Kosolobov and O. Merkurev}
\institute{Ural Federal University, Ekaterinburg, Russia \email{dkosolobov@mail.ru},~\email{o.merkuryev@gmail.com}}
\maketitle
\begin{abstract}
A skeleton Huffman tree is a Huffman tree in which all disjoint maximal perfect subtrees are shrunk into leaves. Skeleton Huffman trees, besides saving storage space, are also used for faster decoding and for speeding up Huffman-shaped wavelet trees. In~2017 Klein et al.\ introduced an optimal skeleton tree: for given symbol frequencies, it has the least number of nodes among all optimal prefix-free code trees (not necessarily Huffman's) with shrunk perfect subtrees. Klein et al.\ described a simple algorithm that, for fixed codeword lengths, finds a skeleton tree with the least number of nodes; with this algorithm one can process each set of optimal codeword lengths to find an optimal skeleton tree. However, there are exponentially many such sets in the worst case. We describe an $\Oh(n^2\log n)$-time algorithm that, given $n$ symbol frequencies, constructs an optimal skeleton tree and its corresponding optimal code.

\end{abstract}

\section{Introduction}

The Huffman code~\cite{Huffman} is one of the most fundamental primitives of data compression. Numerous papers are devoted to Huffman codes and their variations; see the surveys \cite{Abrahams} and~\cite{Moffat}. In this paper we investigate the skeleton Huffman trees, introduced by Klein~\cite{Klein}, which are code trees with all (disjoint) maximal perfect subtrees shrunk into leaves (precise definitions follow). We describe the first polynomial algorithm that, for code weights $w_1, w_2, \ldots, w_n$, constructs the smallest in the number of nodes skeleton tree among all trees formed by optimal prefix-free codes. Klein et al.~\cite{KleinEtAl} called such trees \emph{optimal skeleton trees}.

The idea of the skeleton tree is simple: to determine the length of a given codeword in the input bit stream, the standard decoding algorithm for Huffman codes descends in the tree from the root to a leaf; instead, one can descend to a leaf of the skeleton tree, where the remaining length is uniquely determined by the height of the corresponding shrunk perfect subtree. While presently the decoding is performed by faster table methods~\cite{Moffat} and, in general, the entropy encoding is implemented using superior methods like ANS~\cite{Duda}, there are still important applications for the trees of optimal codes in compressed data structures, where one has to perform the tree descending. For instance, two such applications are in compressed pattern matching~\cite{DaptardarShapira} and in the so-called Huffman-shaped wavelet trees~\cite{GrossiGuptaVitter,MakinenNavarro2}, the basis of FM-indexes~\cite{FerraginaManzini}: the access time to the FM-index might be decreased by skeleton trees in exchange to slower search operations~\cite{BaruchKleinShapira}.

The skeleton trees were initially introduced only for canonical trees (in which the depths of leaves do not decrease when listed from left to right),~as~they~showed good performance in practice~\cite{Klein}. However, as it was noticed in~\cite{KleinEtAl}, the smallest skeleton trees might be induced by neither Huffman nor canonical trees. In order to find optimal skeleton trees, Klein et al.~\cite{KleinEtAl} described a simple algorithm that, for fixed codeword lengths, builds a skeleton tree with the least number of nodes. As a consequence, to find an optimal skeleton tree, one can process each set of optimal codeword lengths using this algorithm. However, there are exponentially many such sets in the worst case \cite{Golomb} and, hence, such algorithm is not polynomial.

To develop a polynomial algorithm, we first prove that it suffices to consider only Huffman codes, not all optimal codes. Then, we investigate properties of Huffman trees resembling the known sibling property~\cite{Gallager}. It turns out that all Huffman trees, for fixed code weights $w_1, w_2, \ldots, w_n$, share a similar layered structure and our dynamic programming $\Oh(n^2\log n)$-time algorithm is based on it. Since normally in practice there are few choices of optimal codeword lengths for given weights, our result is mostly of theoretical value but the found properties, we believe, might be interesting by themselves.

The paper is organized as follows. In Section~\ref{sec:huffman-skeleton} we define all basic concepts and overview some known results. In Section~\ref{sec:layered-struct} we explore the layered structure of Huffman trees that underlies our construction. Section~\ref{sec:algorithm} begins with a simpler cubic algorithm and, then, proceeds to improve its time to $\Oh(n^2\log n)$.

\section{Huffman and Skeleton Trees}
\label{sec:huffman-skeleton}

Throughout the paper, all trees are rooted and binary. A tree is \emph{full} if all its nodes have either two or zero children. The \emph{depth} of a node is the length of the path from the root to the node. A \emph{subtree} rooted at a node is the tree consisting of the node and all its descendants. A \emph{perfect tree} is a full tree in which all leaves have the same depth. A perfect subtree is \emph{maximal} if it is not a subtree of another perfect subtree. Two subtrees are \emph{disjoint} if they do not share common nodes. To \emph{shrink a subtree} is to remove all nodes of the subtree except its root.

Fix $n$ symbols with positive weights $w_1, w_2, \ldots, w_n$ (typically, symbol frequencies). Their \emph{prefix-free code} is a sequence of $n$ binary \emph{codewords} such that no codeword is a prefix of another codeword. Denote by $\ell_1, \ell_2, \ldots, \ell_n$ the codeword lengths. The \emph{quantized source} or \emph{q-source}, as defined in~\cite{FergusonRabinowitz}, is a sequence $q_1, q_2, \ldots, q_m$ such that, for each $\ell$, $q_\ell$ is the number of codewords of length $\ell$ and all lengths are at most $m$. Note that $\sum_{\ell=1}^m q_\ell = n$. It is well known that the lengths and the q-source of any prefix-free code satisfy the Kraft's inequality~\cite{Kraft}:
\begin{equation}\label{eq:kraft}
\sum_{i=1}^n \frac{1}{2^{\ell_i}} = \sum_{\ell=1}^m \frac{q_\ell}{2^\ell} \le 1.
\end{equation}
Conversely, every set of lengths and every q-source satisfying (\ref{eq:kraft}) correspond to a (non-unique) prefix-free code. The code is \emph{optimal} if the sum $\sum_{i=1}^n \ell_i w_i$ is minimal among all possible prefix-free codes. In particular, when $w_1, w_2, \ldots, w_n$ are symbol frequencies in a message, optimal prefix-free codes minimize the total length of the encoded message obtained by the codeword substitution.

We assume that the reader is familiar with Huffman's algorithm~\cite{Huffman} (either its version with heap or queue~\cite{vanLeeuwen}): given positive weights $w_1, w_2, \ldots, w_n$, it builds a full tree whose $n$ leaves are labelled with $w_1, w_2, \ldots, w_n$ and each internal node is labelled with the sum of weights of its children. This tree represents an optimal prefix-free code for the weights: all left edges of the tree are marked with zero, all right edges with one, and each $w_i$ is associated with the codeword written on the corresponding root-leaf path. The obtained tree is not necessarily unique: swapping siblings and arbitrarily breaking ties during the construction (some weights or sums of weights might be equal), one can build many different trees, even exponentially many in some cases as shown by Golomb~\cite{Golomb}; see Fig.~\ref{fig:valid-huffman}. We call all such trees for $w_1, w_2, \ldots, w_n$ \emph{Huffman trees} and they represent \emph{Huffman codes}. The order of siblings and the choice among equal-weighted nodes are the only ties emerging in Huffman's algorithm and leading to different trees. In view of this, the following lemma is straightforward.

\begin{figure}[!htb]
\vskip-2mm
\centering
\subfigure{
\begin{tikzpicture}[font=\tiny,level 1/.style={sibling distance=10mm},level 2/.style={sibling distance=10mm},
                    level 3/.style={sibling distance=8mm},level 4/.style={sibling distance=8mm},level 5/.style={sibling distance=7mm},
                    level distance=0.5cm]
  \node [circle,draw] {27}
    child { node [circle,draw] {18}
      child { node [circle,draw] {9}
        child { node [circle,draw] {6}
          child { node [circle,draw] {3}
            child { node [circle,draw] {2}
              child { node [circle,draw] {1} edge from parent node[above left] {\tiny 0} }
              child { node [circle,draw] {1} edge from parent node[above right] {\tiny 1} }
              edge from parent node[above left] {\tiny 0} }
            child { node [circle,draw] {1} edge from parent node[above right] {\tiny 1} }
            edge from parent node[above left] {\tiny 0} }
          child { node [circle,draw] {3} edge from parent node[above right] {\tiny 1} }
          edge from parent node[above left] {\tiny 0} }
        child { node [circle,draw] {3} edge from parent node[above right] {\tiny 1} }
        edge from parent node[above left] {\tiny 0} }
      child { node [circle,draw] {9} edge from parent node[above right] {\tiny 1} }
      edge from parent node[above left] {\tiny 0} }
    child { node [circle,draw] {9} edge from parent node[above right] {\tiny 1} };
\end{tikzpicture}}\hspace{1cm}
\subfigure{
\begin{tikzpicture}[font=\tiny,level/.style={sibling distance=32mm/#1, level distance=0.7cm}]
  \node [circle,draw] {27}
    child { node [circle,draw] {18}
      child { node [circle,draw] {9} edge from parent node[left] {\tiny 0} }
      child { node [circle,draw] {9} edge from parent node[right] {\tiny 1} }
      edge from parent node[left] {\tiny 0} }
    child { node [circle,draw] {9}
      child { node [circle,draw] {6}
        child { node [circle,draw] {3} edge from parent node[left] {\tiny 0} }
        child { node [circle,draw] {3} edge from parent node[right] {\tiny 1} }
        edge from parent node[left] {\tiny 0} }
      child { node [circle,draw] {3}
        child { node [circle,draw] {2}
          child { node [circle,draw] {1} edge from parent node[left] {\tiny 0} }
          child { node [circle,draw] {1} edge from parent node[right] {\tiny 1} }
          edge from parent node[left] {\tiny 0} }
        child { node [circle,draw] {1} edge from parent node[right] {\tiny 1} }
        edge from parent node[right] {\tiny 1} }
      edge from parent node[right] {\tiny 1} };
\end{tikzpicture}}
\caption{Two Huffman trees for the weights $1,1,1,3,3,9,9$.}\label{fig:valid-huffman}
\vskip-1mm
\end{figure}
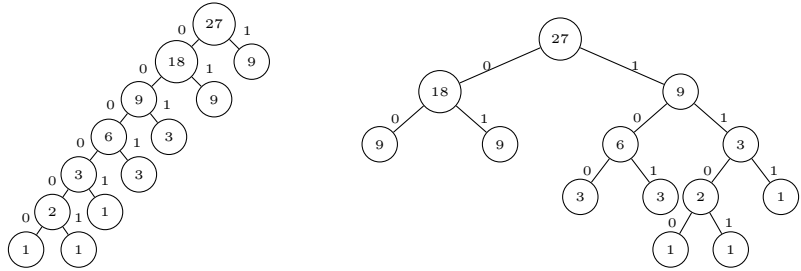

\begin{lemma}\label{lem:valid-huffman}
By swapping subtrees rooted at equal-weighted nodes or swapping siblings in a Huffman tree for $w_1, w_2, \ldots, w_n$, one obtains another Huffman tree and all Huffman trees for $w_1, w_2, \ldots, w_n$ can be reached by these operations.
\end{lemma}

By analogy to Huffman trees, each prefix-free code forms a tree with $n$ leaves labelled with $w_1, w_2, \ldots, w_n$ (and internal nodes labelled with sums of weights of their children). For a given tree, its \emph{skeleton tree}~\cite{Klein} is obtained by choosing all (disjoint) maximal perfect subtrees and then shrinking them. An \emph{optimal skeleton tree}~\cite{KleinEtAl} is a skeleton tree with the least number of nodes among all skeleton trees for the trees formed by optimal prefix-free codes. Fig.~\ref{fig:non-opt-huffman} gives an example from~\cite{KleinEtAl} showing that optimal skeleton trees are not necessarily obtained from Huffman trees: only the left tree is Huffman's, the skeleton trees are drawn in gray, and both codes are optimal (note also that not every optimal code can be obtained by Huffman's algorithm).

\subfigcapmargin = -1cm
\begin{figure}[!htb]
\centering
\subfigure[The Huffman tree.]{
\begin{tikzpicture}[font=\small,level/.style={sibling distance=25mm/#1, level distance=0.8cm}]
  \node [circle,draw,fill=lightgray] {19}
    child { node [circle,draw,fill=lightgray] {8} 
      child { node [circle,draw,fill=lightgray] {4} edge from parent node[left] {\tiny 0}  }
      child { node [circle,draw,fill=lightgray] {4}
        child { node [circle,draw] {2} edge from parent node[left] {\tiny 0} }
        child { node [circle,draw] {2} edge from parent node[right] {\tiny 1} }
        edge from parent node[right] {\tiny 1} }
      edge from parent node[left] {\tiny 0} }
    child { node [circle,draw,fill=lightgray] {11}
      child { node [circle,draw,fill=lightgray] {5} edge from parent node[left] {\tiny 0}  }
      child { node [circle,draw,fill=lightgray] {6}
        child { node [circle,draw] {3} edge from parent node[left] {\tiny 0} }
        child { node [circle,draw] {3} edge from parent node[right] {\tiny 1} }
        edge from parent node[right] {\tiny 1} }
      edge from parent node[right] {\tiny 1} };
\end{tikzpicture}}\hspace{1cm}
\subfigure[The tree inducing optimal skeleton tree.]{
\begin{tikzpicture}[font=\small,level 1/.style={sibling distance=25mm},level 2/.style={sibling distance=15mm},
                    level 3/.style={sibling distance=8mm},level distance=0.8cm]
  \node [circle,draw,fill=lightgray] {19}
    child { node [circle,draw,fill=lightgray] {9}
      child { node [circle,draw] {4} edge from parent node[left] {\tiny 0} }
      child { node [circle,draw] {5} edge from parent node[right] {\tiny 1} }
      edge from parent node[left] {\tiny 0} }
    child { node [circle,draw,fill=lightgray] {10}
      child { node [circle,draw] {4}
        child { node [circle,draw] {2} edge from parent node[left] {\tiny 0} }
        child { node [circle,draw] {2} edge from parent node[right] {\tiny 1} }
        edge from parent node[left] {\tiny 0} }
      child { node [circle,draw] {6}
        child { node [circle,draw] {3} edge from parent node[left] {\tiny 0} }
        child { node [circle,draw] {3} edge from parent node[right] {\tiny 1} }
        edge from parent node[right] {\tiny 1} }
      edge from parent node[right] {\tiny 1}};
\end{tikzpicture}}
\caption{Trees for the weights $2,2,3,3,4,5$; the gray nodes form skeleton trees.}\label{fig:non-opt-huffman}
\end{figure}
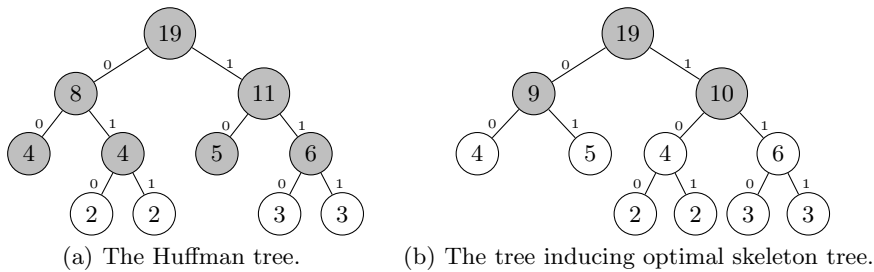

To find an optimal skeleton tree for $w_1, w_2, \ldots, w_n$, Klein et al.~\cite{KleinEtAl} addressed the following problem: given a q-source $q_1, q_2, \ldots, q_n$ of an optimal prefix-free code of size $n$, find a prefix-free code having this q-source whose tree induces a skeleton tree with the least number of nodes. Since $q_\ell$ in optimal codes can be non-zero only for $\ell < n$, we ignore $q_{n+1}, q_{n+2}, \ldots$ in the q-source specifying only $q_1,q_2,\ldots,q_n$ (though $q_n$ can be omitted too). Let us sketch the solution of~\cite{KleinEtAl}.

Denote by $\pop(x)$ the number of ones in the binary representation of integer $x$. Since $q_1, q_2, \ldots, q_n$ is the q-source of an optimal code, (\ref{eq:kraft}) is an equality rather than inequality: $\sum_{\ell=1}^n \frac{q_\ell}{2^\ell} = 1$. Hence, any resulting code tree is full and so is its skeleton tree. Consider such skeleton tree. As it is full, the smallest such tree has the least number of leaves. Each of its leaves is a shrunk perfect subtree (possibly, one-node subtree). Split all $q_\ell$ depth-$\ell$ leaves of the corresponding code tree into $r$ subsets, each of which is a power of two in size and consists of all leaves of a shrunk subtree. Then, $r \ge k$ if $q_\ell = \sum_{i=1}^k 2^{m_i}$ for some $m_1 > m_2 > \cdots > m_k$, i.e., $r \ge \pop(q_\ell)$. Therefore, the skeleton tree has at least $\sum_{\ell=1}^n \pop(q_\ell)$ leaves.

The bound is attainable. Shrinking a perfect subtree with $2^m$ depth-$\ell$ leaves, we decrease $q_\ell$ by $2^m$ and increment $q_{\ell - m}$ by $1$, which does not affect the sum~(\ref{eq:kraft}) since $\frac{1}{2^{\ell - m}} = \frac{2^m}{2^\ell}$. Based on this, we initialize with zeros some $q'_1, q'_2, \ldots, q'_n$ and, for $\ell \in \{1,2,\ldots,n\}$, increment $q'_{\ell - m_1}, \ldots, q'_{\ell - m_k}$, where $q_\ell = \sum_{i=1}^k 2^{m_i}$ is the binary representation of $q_\ell$. In the end, $q'_1, q'_2, \ldots, q'_n$ satisfy~(\ref{eq:kraft}) and $\sum_{\ell=1}^n q'_\ell = \sum_{\ell=1}^n \pop(q_\ell)$. By a standard method, we build a full tree having, for each $\ell$, $q'_\ell$ depth-$\ell$ leaves and it is precisely the sought skeleton tree. By appropriately ``expanding'' its leaves into perfect subtrees, one can construct the corresponding code tree (see details in~\cite{KleinEtAl}). Thus, Klein et al.\ proved the following lemma.

\begin{lemma}
Let $q_1, q_2, \ldots, q_n$ be a q-source of a size-$n$ code such that $\sum_{\ell=1}^n \frac{q_\ell}{2^\ell}{=}1$. The smallest skeleton tree in the number of nodes for a tree of a prefix-free code having this q-source has $\sum_{\ell=1}^n \pop(q_\ell)$ leaves and, thus, $2\sum_{\ell=1}^n \pop(q_\ell) - 1$ nodes.\label{lem:klein-algorithm}
\end{lemma}

Lemma~\ref{lem:klein-algorithm} implies that, for weights admitting a Huffman tree of height $h$, any optimal skeleton tree has at most $2 h \log_2 n$ nodes. In particular, as noted in~\cite{Klein}, the skeleton tree has $\Oh(\log^2 n)$ nodes if a Huffman tree is of height $\Oh(\log n)$.

By Lemma~\ref{lem:klein-algorithm}, one can find an optimal skeleton tree for weights $w_1, w_2, \ldots, w_n$ by searching the minimum of $\sum_{\ell=1}^n \pop(q_\ell)$ among all q-sources $q_1, q_2, \ldots, q_n$ yielding optimal codes. Such algorithm is exponential in the worst case as it follows from~\cite{Golomb}. We are to develop a polynomial-time algorithm for this problem.

\section{Layered Structure of Huffman Trees}
\label{sec:layered-struct}

The following property of monotonicity in trees of optimal codes is well known.

\begin{lemma}
If the weights of nodes $u$ and $u'$ in the tree of an optimal prefix-free code are $w$ and $w'$, and $w > w'$, then the depth of $u$ is at most that of $u'$.\label{lem:heap-like}
\end{lemma}

Due to Lemma~\ref{lem:klein-algorithm}, our goal is to find a q-source $q_1, q_2, \ldots, q_n$ yielding an optimal prefix-free code for $w_1, w_2, \ldots, w_n$ that minimizes $\sum_{\ell=1}^n \pop(q_\ell)$. Huffman trees represent optimal codes; are there other optimal codes whose q-sources we have to consider? The following lemma answers this question in negative.

\begin{lemma}
Let $q_1, q_2, \ldots, q_n$ be a q-source of an optimal prefix-free code for weights $w_1, w_2, \ldots, w_n$. Then, there is a Huffman code having the same q-source.\label{lem:huffman-only}
\end{lemma}
\begin{proof}
The proof is by induction on $n$. The case $n \le 2$ is trivial, so assume $n > 2$. Let $w_1 \ge w_2 \ge \cdots \ge w_n$. Consider an optimal prefix-free code for $w_1, w_2, \ldots, w_n$ with the q-source $q_1, q_2, \ldots, q_n$ and lengths $\ell_1, \ell_2, \ldots, \ell_n$. Let $\ell = \max\{\ell_1, \ell_2, \ldots, \ell_n\}$. The tree for the code is full and, due to Lemma~\ref{lem:heap-like}, has depth-$\ell$ leaves with weights $w_{n-1}$ and $w_n$. By swapping depth-$\ell$ leaves, we can make $w_{n-1}$ and $w_n$ siblings. Shrinking these siblings into one leaf of weight $w' = w_{n-1} + w_n$, we obtain the tree for a code with the q-source $q_1, q_2, \ldots, q_{\ell-2}, q_{\ell-1} + 1, q_{\ell} - 2$. The tree represents an optimal prefix-free code for the weights $w_1, w_2, \ldots, w_{n-2}, w'$: its total cost is $\sum_{i=1}^{n} \ell_i w_i - w'$ and any prefix-free code with a smaller cost would induce a code of cost smaller than $\sum_{i=1}^n \ell_i w_i$ for $w_1, w_2, \ldots, w_n$ (which is impossible) by expanding a leaf of weight $w'$ in its tree into two leaves of weights $w_{n-1}$ and $w_n$. Since the code is optimal, by Lemma~\ref{lem:heap-like}, the smallest $q_\ell - 2$ weights in the set $\{w_1, w_2, \ldots, w_{n-2}, w'\}$ mark all depth-$\ell$ leaves in the tree and the next smallest $q_{\ell-1} + 1$ weights mark all depth-$(\ell{-}1)$ leaves; the weight $w'$ is in depth $\ell - 1$ and, so, is in the second group. By the inductive hypothesis, there is a Huffman tree for the weights $w_1, w_2, \ldots, w_{n-2}, w'$ with the q-source $q_1, q_2, \ldots, q_{\ell-2}, q_{\ell-1} + 1, q_{\ell} - 2$. Again by Lemma~\ref{lem:heap-like}, the smallest  $q_\ell - 2$ weights in the set $\{w_1, w_2, \ldots, w_{n-2}, w'\}$ mark all depth-$\ell$ leaves in the Huffman tree and the next smallest $q_{\ell-1} + 1$ ($w'$ among them) mark depth $\ell - 1$. The leaf of weight $w'$ can be expanded into two leaves of weights $w_{n-1}$ and $w_n$, thus producing a Huffman tree for the weights $w_1, w_2, \ldots, w_n$ (since Huffman's algorithm by its first step unites $w_{n-1}$ and $w_n$ into $w'$) with the q-source $q_1, q_2, \ldots, q_n$.\qed
\end{proof}

By Lemma~\ref{lem:huffman-only}, it suffices to consider only q-sources of Huffman codes. Instead of processing them all, we develop a different approach. Consider a Huffman tree and all its nodes with a given depth $\ell$. Denote by $v_1, v_2, \ldots, v_i$ all distinct weights of these nodes such that $v_1 > v_2 > \cdots > v_i$. For $1 \le j \le i$, let $h_j$ be the number of depth-$\ell$ nodes of weight $v_j$, and $k_j$ be the number of depth-$\ell$ leaves of weight $v_j$ (so that $k_j \le h_j$). Lemma~\ref{lem:heap-like} implies that all nodes of weight $v$ such that $v_1 > v > v_i$ (if any) have depth $\ell$, i.e., are entirely inside the depth-$\ell$ ``layer''. Based on this observation, one can try to tackle the problem using the dynamic programming that, given parameters $\ell, v_1, h_1, k_1, v_i, h_i, k_i$, computes the minimal sum $\sum_{\ell' \ge \ell} \pop(q_{\ell'})$ among all q-sources induced by Huffman trees in which the depth-$\ell$ ``layer'' is compatible with $v_1, h_1, k_1, v_i, h_i, k_i$. The main challenge in this approach is to figure out somehow all possible configurations of the depth-$(\ell{+}1)$ layer compatible with the depth-$\ell$ parameters. In what follows, we prove a number of structural lemmas that resolve such issues and simplify the method; in particular, it turns out, the parameters $\ell, v_i, h_i, k_i$ can be omitted.

\begin{lemma}
The depths of equal-weight nodes in Huffman tree differ by at most~$1$.\label{lem:class-per-layers}%
\end{lemma}
\begin{proof}
Since all weights are positive and the tree is full, the parent of any node of weight $w$ has weight larger than $w$. Then, by Lemma~\ref{lem:heap-like}, the depths of all nodes of weight $w$ are at least the depth of that parent. Hence, the result follows.\qed
\end{proof}

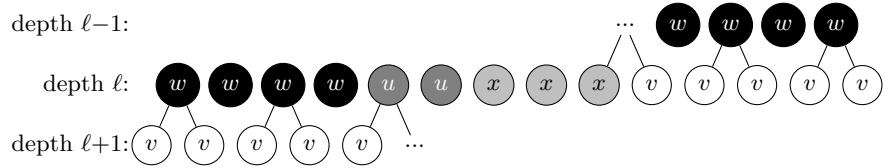
\begin{figure}[!htb]
\begin{tikzpicture}[font=\small,level/.style={sibling distance=0.7cm, level distance=0.8cm}]
  \node  at (-0.73,0.8) {depth $\ell{-}1$:};
  \node  at (-0.52,0) {depth $\ell$:};
  \node  at (-0.73,-0.8) {depth $\ell{+}1$:};
  \node [circle,draw,fill=black,text=white]  at (0.7,0) {$w$}
    child { node [circle,draw,fill=white] {$v$} }
    child { node [circle,draw,fill=white] {$v$} };
  \node [circle,draw,fill=black,text=white]  at (1.4,0) {$w$};
  \node [circle,draw,fill=black,text=white]  at (2.1,0) {$w$}
    child { node [circle,draw,fill=white] {$v$} }
    child { node [circle,draw,fill=white] {$v$} };
  \node [circle,draw,fill=black,text=white]  at (2.8,0) {$w$};
  \node [circle,draw,fill=gray,text=white]  at (3.5,0) {$u$}
    child { node [circle,draw,fill=white] {$v$} }
    child { node [fill=white] {$...$} };
  \node [circle,draw,fill=gray,text=white]  at (4.2,0) {$u$};
  \node [circle,draw,fill=lightgray]  at (4.9,0) {$x$};
  \node [circle,draw,fill=lightgray]  at (5.6,0) {$x$};
  \node [fill=white]  at (6.65,0.8) {$...$}
    child { node [circle,draw,fill=lightgray] {$x$} }
    child { node [circle,draw,fill=white] {$v$} };
  \node [circle,draw,fill=black,text=white]  at (7.35,0.8) {$w$};
  \node [circle,draw,fill=black,text=white]  at (8.05,0.8) {$w$}
    child { node [circle,draw,fill=white] {$v$} }
    child { node [circle,draw,fill=white] {$v$} };
  \node [circle,draw,fill=black,text=white]  at (8.75,0.8) {$w$};
  \node [circle,draw,fill=black,text=white]  at (9.45,0.8) {$w$}
    child { node [circle,draw,fill=white] {$v$} }
    child { node [circle,draw,fill=white] {$v$} };
\end{tikzpicture}
\caption{A ``layer'': all nodes whose weights appear on depth $\ell$ ($w > u > x > v$).}\label{fig:layer}
\end{figure}

Due to Lemmas~\ref{lem:heap-like} and~\ref{lem:class-per-layers}, for each $\ell$, the depth-$\ell$ nodes with the largest and smallest weights (among all depth-$\ell$ nodes) can have equal-weighted nodes with depths $\ell - 1$ and $\ell + 1$, respectively, but all other depth-$\ell$ nodes cannot have equal-weighted nodes of different depths. See Fig.~\ref{fig:layer}. We are to show that this layout is quite stable: in all Huffman trees the nodes of a given weight are the same in number and have the same children, and the heaviest nodes in each layer, in a sense, determine the lightest nodes.

Since all node weights appear in the queue of Huffman's algorithm in non-decreasing order, the following claim is immediate.

\begin{lemma}
In a Huffman tree, for any $v \le v'$ and $w \le w'$, if a pair of siblings has weights $v$ and $v'$, and a different pair of siblings has weights $w$ and $w'$, then we have either $w' \le v$ or $v' \le w$.
\label{lem:sibling-order}
\end{lemma}

\begin{lemma}
One of the following (mutually exclusive) alternatives holds for the set of all nodes of a given weight $w$ in any Huffman tree for $w_1, w_2, \ldots, w_n$:
\begin{enumerate}
\item all the nodes are leaves;
\item all the nodes except one are leaves and this one node has two children with weights $w'$ and $w''$ such that $w' < w''$ and $w' + w'' = w$;
\item all the nodes except $m$ of them, for some $m > 0$, are leaves and each of these $m$ nodes has two children with weights $w/2$.
\end{enumerate}
In different Huffman trees the sets of all nodes of weight $w$ have the same size and the same alternative holds for them
with the same $w'$, $w''$, and $m$, respectively.\label{lem:class-struct}
\end{lemma}
\begin{proof}
It suffices to prove that if two nodes of weight $w$ are not leaves, then their children have weights $w/2$. Suppose, to the contrary, that the children have weights, respectively, $w'$, $w''$, and $v'$, $v''$ such that $w' < w''$ and $v' \le v''$. But then $v' < w''$ and $w' < v''$, which is impossible by Lemma~\ref{lem:sibling-order}.

As the swapping operations of Lemma~\ref{lem:valid-huffman} do not change the number of nodes of weight $w$ and the weights of their children, the last part of the lemma (the preservation of the alternative with given $w'$, $w''$ or $m$) is straightforward.\qed
\end{proof}

The heaviest layer nodes determine the lightest ones as follows.

\begin{lemma}
In a Huffman tree, for $\ell \ge 0$, choose (if any) the largest $w$ among all weights of depth-$\ell$ nodes such that there is a non-leaf node of weight $w$ (maybe not of depth $\ell$). Let $t$ be a node of weight $w$ with children of weights $w'$ and $w''$ such that $w' \le w''$. Then, the weights of depth-$\ell$ nodes are at least $w'$, the weights of depth-$(\ell{+}1)$ nodes are at most $w''$, and the next alternatives are possible:
\begin{enumerate}
\item $w' < w''$, the depth of $t$ is $\ell$, and only one node of depth $\ell+1$ has weight $w''$;\label{en:t-on-level}
\item $w' < w''$, the depth of $t$ is $\ell - 1$, and only one node of depth $\ell$ has weight $w'$;\label{en:t-over-level}
\item $w' = w'' = w/2$, for some $m$, exactly $m$ non-leaf nodes of depth $\ell$ have weight $w$, and exactly $2m + \delta$ nodes of depth $\ell+1$ have weight $w/2$, where $\delta = 1$ if some node of weight $w/2$ has sibling of smaller weight, and $\delta = 0$ otherwise.\label{en:equal-weights}
\end{enumerate}\label{lem:layer-struct}
\end{lemma}
\begin{proof}
Due to Lemma~\ref{lem:sibling-order}, any node of depth $\ell + 1$ with weight larger than $w''$ has a sibling of weight at least $w''$. Hence, their parent, whose depth is $\ell$, has weight larger than $2w'' \ge w$, which contradicts the choice of $w$. Analogously, any node of depth $\ell$ with weight less than $w'$ has a sibling of weight at most $w'$ and, hence, their parent, whose depth is $\ell-1$, has weight smaller than $2w' \le w$, which is impossible by Lemma~\ref{lem:heap-like} since there is a node of depth $\ell$ with weight $w$. Thus, the first part of the lemma is proved. Let us consider $t$ and its alternatives.

Let $w' < w''$. By Lemma~\ref{lem:class-per-layers}, the depth of $t$ differs from $\ell$ by at most one. Since the weight of $t$ is larger than $w''$, $t$ cannot have depth $\ell + 1$. Suppose the depth of $t$ is $\ell$. By Lemma~\ref{lem:sibling-order}, any node $v$ of depth $\ell + 1$ and weight $w''$ whose parent is not $t$ has a sibling of weight at least $w''$. Hence, its parent, whose depth is $\ell$, is of weight at least $2w'' > w$, which contradicts the choice of $w$. Thus, there is no such $v$. Suppose the depth of $t$ is $\ell - 1$. By Lemma~\ref{lem:sibling-order}, any node $v$ of depth $\ell$ and weight $w'$ whose parent is not $t$ has a sibling of weight at most $w'$ and their parent, whose depth is $\ell - 1$, has weight at most $2w' < w$, which contradicts Lemma~\ref{lem:heap-like} since we have a depth-$\ell$ node of weight $w$. Thus, there is no such $v$.

Let $w' = w'' = w/2$. Lemma~\ref{lem:sibling-order} implies that at most one node of weight $w/2$ can have a sibling of smaller weight. Suppose such node $v$ exists. The weight of its parent is less than $w$ but larger than $w/2$. Since all nodes of depth $\ell$ have weights at most $w/2$, the depth of the parent is at most $\ell$ and, by Lemma~\ref{lem:heap-like}, at least $\ell$. Hence, $v$ has depth $\ell + 1$ and, therefore, the total number of nodes of depth $\ell + 1$ with weights $w/2$ is $2m + 1$. The remaining details are obvious.\qed
\end{proof}

The layout described in Lemmas~\ref{lem:class-struct} and~\ref{lem:layer-struct} can be seen as a more detailed view on the well-known sibling property~\cite{Gallager}: listing all node weights layer by layer upwards, one can obtain a non-decreasing list in which all siblings are adjacent.

\section{Algorithm}
\label{sec:algorithm}

Our scheme is as follows. Knowing the heaviest nodes in a given depth-$\ell$ layer of a Huffman tree and which of them are leaves, one can deduce from Lemma~\ref{lem:layer-struct} the lightest layer nodes. Then, by Lemma~\ref{lem:heap-like}, the heaviest nodes of the depth-$(\ell{+}1)$ layer either immediately precede the lightest nodes of depth $\ell$ in the weight-sorted order or have the same weight; in any case they can be determined and the depth-$(\ell{+}1)$ layer can be recursively processed. Let us elaborate details.

In the beginning, the algorithm constructs a Huffman tree for $w_1, w_2, \ldots, w_n$ and splits all nodes into classes $c_1, c_2, \ldots, c_r$ according to their weights: all nodes of class $c_u$ are of the same weight denoted $\hat{w}_u$ and $\hat{w}_1 > \hat{w}_2 > \cdots > \hat{w}_r$. In addition to $\hat{w}_u$, the algorithm calculates for each class $c_u$ the following attributes:
\begin{itemize}
\item the total number of nodes in $c_u$, denoted $|c_u|$;
\item the number of leaves in $c_u$, denoted $\lambda_u$;
\item provided $\lambda_u \ne |c_u|$ (i.e., not all class nodes are leaves), the weights of two children of an internal node with weight $\hat{w}_u$ (by Lemma~\ref{lem:class-struct}, the choice of the node is not important), denoted $w'_u$ and $w''_u$ and such that $w'_u \le w''_u$;
\item the number $\delta_u$ such that $\delta_u = 1$ if $\lambda_u \ne |c_u|$ and a node of weight $w'_u$ has a sibling of weight ${<}w'_u$, and $\delta_u = 0$ otherwise ($\delta_u$ serves as $\delta$ from Lemma~\ref{lem:layer-struct});
\item the number $t_u$ such that $t_u = 0$ if $\lambda_t = |c_t|$ for all $t \ge u$, and $t_u = \min\{t \ge u \colon \lambda_t \ne |c_t|\}$ otherwise ($t_u$ is used to identify a node $t$ as in Lemma~\ref{lem:layer-struct}).
\end{itemize}
It is straightforward that the attributes, for all classes, can be computed in $\Oh(n \log n)$ overall time. By Lemma~\ref{lem:class-struct}, the class information (i.e., class weights and attributes) is the same in all possible Huffman trees for $w_1, w_2, \ldots, w_n$.

We use a dynamic programming with parameters $c_u$, $h$, $k$: $c_u$ is the class to which the heaviest nodes in a depth-$\ell$ layer of a Huffman tree belong, $h \le |c_u|$ is the number of nodes from $c_u$ with depth $\ell$, and $k \le h$ is the number of leaves from $c_u$ with depth $\ell$ (note that $\ell$ is not known). Informally, the algorithm recursively computes from $c_u, h, k$ the minimal sum $\sum_{\ell' \ge \ell} \pop(q_{\ell'})$, where $q_{\ell'}$ is the number of leaves with depth $\ell'$, among all Huffman trees that, for some $\ell$, have a depth-$\ell$ layer compatible with the parameters $c_u, h, k$. Using the class information, one can deduce from $c_u$, $h$, $k$ the weight $\hat{w}_{u'}$ of the heaviest nodes of the lower depth-$(\ell{+}1)$ layer, and the number $h'$ of the nodes of weight $\hat{w}_{u'}$ with depth $\ell + 1$: if $t_u = 0$, our layer is lowest and we simply return $\pop(k + \sum_{i=u+1}^r \lambda_i)$; otherwise, we denote $t = t_u$ and consider several cases according to Lemma~\ref{lem:layer-struct}:
\begin{enumerate}
\item $w'_{t} < w''_{t}$ and either $t \ne u$ or $k \ne h$; then $\hat{w}_{u'} = w''_{t}$ and $h' = 1$;
\item $w'_{t} < w''_{t}$, $t = u$, and $k = h$ (i.e., as Lemmas~\ref{lem:class-struct} and~\ref{lem:layer-struct} imply, the only non-leaf node of weight $\hat{w}_t$ has depth $\ell - 1$); then either $\hat{w}_{u'} = w'_{t}$ and $h' = |c_{u'}| - 1$ if there are at least two nodes of weight $w'_{t}$, or, otherwise, $\hat{w}_{u'}$ is the weight immediately preceding $w'_{t}$ in the sorted set of all weights and $h' = |c_{u'}|$;
\item[3.1] $w'_{t} = w''_{t}$, $t = u$, and $2(h - k) + \delta_u \ne 0$; then $\hat{w}_{u'} = \hat{w}_{u}/2$ and $h' = 2(h - k) + \delta_{u}$;
\item[3.2] $w'_{t} = w''_{t}$, $t = u$, and $2(h - k) + \delta_u = 0$; then $\hat{w}_{u'}$ is the weight immediately preceding $\hat{w}_{u}/2$ in the sorted set of all weights and $h' = |c_{u'}|$;
\item[3.3] $w'_{t} = w''_{t}$ and $t \ne u$; then $\hat{w}_{u'} = \hat{w}_t / 2$ ($=w'_{t}$) and $h' = 2(|c_{t}| - \lambda_{t}) + \delta_{t}$.
\end{enumerate}

Thus, the parameters $c_u$, $h$, $k$ uniquely determine the weight $\hat{w}_{u'}$ (and, hence, the class $c_{u'}$) of the heaviest nodes of the subsequent depth-$(\ell{+}1)$ layer and the number $h'$ of nodes of weight $\hat{w}_{u'}$ with depth $\ell + 1$ (recall that $\ell$ denotes the depth of the ``$(c_u, h, k)$-layer'' and is unknown). The number of leaves of weight $\hat{w}_{u'}$ with depth $\ell + 1$, denoted $k'$, is a ``free'' parameter of the recursion. We loop trough all possible $k'$ and, thus, compute the minimum of the sums $\sum_{\ell' \ge \ell} q_{\ell'}$.

The parameter $k'$ is not arbitrary: it cannot exceed neither $h'$ nor $\lambda_{u'}$, and cannot make $\lambda_{u'} - k'$ (the number of leaves from $c_{u'}$ with depth $\ell$) larger than $|c_{u'}| - h'$ (the total number of nodes from $c_{u'}$ with depth $\ell$). Processing all $k'$ subject to these restrictions, we compute the answer, $f(c_u, h, k)$, as follows:
\begin{equation}\label{eq:main-recursion}
f(c_u, h, k){\,=}\min_{\substack{0\le k'\le\min\{h', \lambda_{u'}\}\\\mbox{\tiny and}\\ \lambda_{u'} - |c_{u'}| + h' \le k'}}\left\{f(c_{u'}, h', k') + \mathop{\pop}\!\!\left(k{+}\lambda_{u'}{-}k'{+}\!\!\!\sum_{i=u+1}^{u'-1} \lambda_i\right)\right\}.
\end{equation}
Thus, if (\ref{eq:main-recursion}) is indeed a solution, $f(c_1, 1, 0)$ gives the answer for the problem: the minimal $\sum_{\ell=1}^n \pop(q_\ell)$ for all q-sources $q_1, q_2, \ldots, q_n$ of optimal prefix-free codes; the restriction to Huffman trees in the design of the recursion is justified by Lemma~\ref{lem:huffman-only}. Using memoization instead of the mere recursion in~(\ref{eq:main-recursion}), one obtains a polynomial time algorithm. The optimal skeleton tree itself is constructed by applying the algorithm of Klein et al.~\cite{KleinEtAl} to a q-source $q_1, q_2, \ldots, q_n$ achieving the minimum, which can be found by the standard backtacking technique.

It is not immediate, however, that every choice of the parameter $k'$ in~(\ref{eq:main-recursion}) yields a ``layered structure'' (designated by recursive calls) corresponding to a valid Huffman tree for $w_1, w_2, \ldots, w_n$; in principle, one could imagine a situation when the assignments of $c_{u'}, h', k'$ leading to the minimum in~(\ref{eq:main-recursion}) do not correspond to any Huffman tree at all. Fortunately, this is not the case. Let us prove that whenever a triple $(c_u, h, k)$ is reached on an $(\ell{+}1)$st level of the recursion~(\ref{eq:main-recursion}) starting from $f(c_1, 1, 0)$ (so that the triple $(c_1, 1, 0)$ is on the first level), it is guaranteed that there exists a Huffman tree for $w_1, w_2, \ldots, w_n$ in which the heaviest nodes of the depth-$\ell$ layer are from the class $c_u$, exactly $h$ nodes of weight $\hat{w}_u$ have depth $\ell$, exactly $k$ nodes of these $h$ nodes are leaves, and every (higher) depth-$\ell'$ layer, for $\ell' < \ell$, in the tree is analogously compatible with a corresponding triple $(c_{u'}, h', k')$ on level $\ell'$ leading to $(c_u, h, k)$ in the recursion.

Obviously, the only node of depth zero (the root) is from the class $c_1$ with $|c_1| = 1$ in all possible Huffman trees. Suppose that a triple $(c_u, h, k)$ is reached in~(\ref{eq:main-recursion}) on level $\ell + 1$ starting from $f(c_1, 1, 0)$ and there is a Huffman tree $\mathcal{T}$ in which the heaviest nodes with depth $\ell$ are from $c_u$, exactly $h$ of the class $c_u$ nodes have depth $\ell$, $k$ of these $h$ nodes are leaves, and all (higher) depth-$\ell'$ layers, for $\ell' < \ell$, are compatible with their corresponding triples from the call stack of the recursion. The argument above shows that the triple $(c_u, h, k)$ uniquely determines analogous parameters $c_{u'}$ and $h'$ (but not $k'$) for the layer of depth $\ell + 1$ in $\mathcal{T}$. Fix an arbitrary integer $k'$ satisfying the conditions imposed by the minimum in~(\ref{eq:main-recursion}). It is straightforward that by swapping subtrees rooted at weight-$\hat{w}_{u'}$ nodes of depths $\ell$ and $\ell + 1$, one can transform $\mathcal{T}$ into a tree having exactly $k'$ leaves from $c_{u'}$ of depth $\ell + 1$. By Lemma~\ref{lem:valid-huffman}, the obtained tree is a Huffman tree too and the transformations do not affect the nodes of class $c_u$ and nodes with depths less than $\ell$. Therefore, the claim is proved and the correctness of the algorithm immediately follows from it.

Let us estimate the running time. The sum $\sum_{i=u+1}^{u'-1} \lambda_i$ in (\ref{eq:main-recursion}) can be calculated in $\Oh(1)$ time if one has precomputed sums $\sum_{i=1}^{v} \lambda_i$, for all $v \le r$. All other parameters in (\ref{eq:main-recursion}), except $k'$, are precomputed and accessible in $\Oh(1)$ time. Since $h \le |c_u|$ and $\sum_{i=1}^r |c_i| = 2n-1$ (the number of nodes in any full tree with $n$ leaves is $2n - 1$), the number of different pairs $(c_u, h)$ is $\Oh(n)$. Hence, the number of triples $(c_u, h, k)$ is $\Oh(n^2)$. For each of the triples, the algorithm runs through at most $\Oh(n)$ appropriate $k'$s in~(\ref{eq:main-recursion}). Therefore, the overall running time is $\Oh(n^3)$ and we have proved the following theorem.

\begin{theorem}
An optimal skeleton tree for positive weights $w_1, w_2, \ldots, w_n$ can be constructed in $\Oh(n^3)$ time.\label{thm:cubic-algorithm}
\end{theorem}

The table used by the algorithm of Theorem~\ref{thm:cubic-algorithm} is of size $\Oh(n^2)$. The main slowdown making the running time cubic is in the processing of $\Oh(n)$ possible parameters $k'$ in~(\ref{eq:main-recursion}). We are to optimize this aspect of the algorithm.

For $i, j \in \mathbb{Z}$, denote $[i..j] = \{i, i+1, \ldots, j\}$. Given parameters $c_u, h, k$, the conditions under the minimum in~(\ref{eq:main-recursion}) determine a range $[k'_{min}..k'_{max}]$ of $k'$ that should be processed. As a first attempt for optimization, one might try to pass in the function $f$ not one $k'$ but the whole range $[k'_{min}..k'_{max}]$, as if the function looked like $f(c_{u'}, h', [k'_{min} .. k'_{max}])$. However, in this case it is not clear how to take into account the effects of the choice of $k' \in [k'_{min}.. k'_{max}]$ on higher (with smaller depth) layers, namely, on the following summand from~(\ref{eq:main-recursion}):
\begin{equation}
\mathop{\pop}\left(k + \lambda_{u'} - k' + \sum_{i=u+1}^{u'-1} \lambda_i\right).\label{eq:summand}
\end{equation}
It turns out that the range $[k'_{min}..k'_{max}]$ can be split into $\Oh(\log n)$ disjoint subranges so that, for each subrange, the effect of the choice of $k'$ on higher layers is well determined. The splitting is based on the following lemma.

\begin{lemma}
In $\Oh(\log b)$ time one can split any range $[a..b{-}1]$ into subranges $[i_0..i_1{-}1]$, $[i_1..i_2{-}1], \ldots, [i_{m-1}..i_m{-}1]$ such that $a = i_0 < i_1 < \cdots < i_m = b$, $m \le 2\log_2 b$, and for each $j \in [0..m{-}1]$, the length of the subrange $[i_j .. i_{j+1}{-}1]$ is a power of two and $\pop(i_{j} + x) = \pop(i_{j}) + \pop(x)$ whenever $0 \le x < i_{j+1} - i_{j}$.\label{lem:pop-subranges}
\end{lemma}
\begin{proof}
We use the operation $\mathsf{ctz}(x)$ counting trailing zeros in the binary representation of $x$.\footnote{If not supported, it can be implemented using a precomputed table of size $\Oh(n)$ \cite{FredmanWillard}.} The subranges are generated iteratively: given $i_j$ (initially $i_0 = a$), we compute $d_j = \max\{d \le \mathsf{ctz}(i_j) \colon i_j + 2^d \le b\}$ decrementing $d = \mathsf{ctz}(i_j)$ until $i_j + 2^d \le b$, and put $i_{j+1} = i_j + 2^{d_j}$; we stop when $i_{j+1} = b$. Since $\mathsf{ctz}(x + 2^d) > d$ for $d = \mathsf{ctz}(x)$, the numbers $d_j$ first increase on each iteration ($d_0 < d_1 < \cdots$), then decrease ($\cdots > d_{m-2} > d_{m-1}$), but always $d_j \le \log_2 b$. Therefore, at most $2\log_2 b$ subranges are created. Finally, since each $i_j$, by construction, has at least $d_j$ trailing zeros, we have $\pop(i_j + x) = \pop(i_j) + \pop(x)$ for $0 \le x < 2^{d_j}$.\qed
\end{proof}

Suppose that we compute $f(c_u, h, k)$ using the recursion~(\ref{eq:main-recursion}). The range $[k'_{min} .. k'_{max}]$ for $k'$ determined by the conditions under the minimum in~(\ref{eq:main-recursion}) is bijectively mapped onto a contiguous range of arguments for $\pop$ in~(\ref{eq:summand}), namely, onto $[\Delta - k'_{max} .. \Delta - k'_{min}]$, where $\Delta = k + \lambda_{u'} + \sum_{i=u+1}^{u'-1} \lambda_i$. According to Lemma~\ref{lem:pop-subranges}, in $\Oh(\log n)$ time one can split the range $[k'_{min}..k'_{max}]$ into $\Oh(\log n)$ disjoint subranges $[i_0{+}1..i_1], [i_1{+}1..i_2], \ldots, [i_{m-1}{+}1..i_m]$, for some $i_1 < i_2 < \cdots < i_m$, such that, for any $j \in [1..m]$, the length of $[i_{j-1}{+}1..i_j]$ is a power of two, denoted by $2^{d_j}$, and one has $\pop(\Delta - k') = \pop(\Delta - i_j) + \pop(i_j - k')$ whenever $i_j - 2^{d_j} < k' \le i_j$. Therefore, the minimum in~(\ref{eq:main-recursion}) restricted to any such subrange $[i_{j-1}{+}1..i_j]$ can be transformed as follows:
\begin{multline}\label{eq:f-to-g}
\min\{f(c_{u'}, h', k') + \pop(\Delta - k') \colon i_{j-1} < k' \le i_j\} =\\
 \pop(\Delta - i_j) + \min\{f(c_{u'}, h', i_j - s) + \pop(s) \colon 0 \le s < 2^{d_j}\}.
\end{multline}
In order to compute the minimum of the values $f(c_{u'}, h', i_j - s) + \pop(s)$, we devise a new function $g(c_u, h, k_{max}, d)$ that is similar to $f$ but, instead of simply computing $f(c_u, h, k_{max})$, it finds the minimum of $f(c_u, h, k_{max} - s) + \pop(s)$ for all $s \in [0 .. 2^d{-}1]$. In other words, the function $g$ computes the minimum of $f(c_u, h, k)$ in the range $k_{max}-2^d < k \le k_{max}$ taking into account the effect on the higher layer (something that $f$ does not care). Using~(\ref{eq:f-to-g}), one can calculate (\ref{eq:main-recursion}) by finding the following minimum: $\min\{\pop(\Delta - i_j) + g(c_{u'}, h', i_j, d_j) \colon j \in [1..m]\}$, which contains $\Oh(\log n)$ invocations of $g$ since $m = \Oh(\log n)$ due to Lemma~\ref{lem:pop-subranges}. The implementation of the function $g$ itself is rather straightforward:
\begin{equation}\label{eq:main-recursion2}
\begin{array}{l}
g(c_u, h, k_{max}, 0) = f(c_u, h, k_{max}),\\
g(c_u, h, k_{max}, d) = \min\left\{\begin{array}{l}g(c_u, h, k_{max}, d-1),\\ g(c_u, h, k_{max} - 2^{d-1}, d-1) + 1.\end{array}\right.
\end{array}
\end{equation}
While computing $g(c_u, h, k_{max}, d)$ with $d > 0$, we choose whether to set the $d$th bit of $s$ in order to minimize $f(c_u, h, k_{max}\!- s) + \pop(s)$: if set, it affects $\pop(s)$ increasing it by one and we still can choose $d-1$ lower bits, hence, the answer is the same as for $g(c_u, h, k_{max}\!-2^{d-1}, d-1) + 1$; if not set, the answer is the same as for $g(c_u, h, k_{max}, d-1)$.  Thus, $g(c_u, h, k_{max}, d)$ implemented as in~(\ref{eq:main-recursion2}) indeed computes the required $\min\{f(c_u, h, k_{max}\!- s) + \pop(s) \colon 0 \le s < 2^d\}$.

The correctness of the algorithm should be clear by this point. Let us estimate the running time. The splitting process in the described implementation of $f(c_u, h, k)$ takes $\Oh(\log n)$ time, by Lemma~\ref{lem:pop-subranges}, which is a significant improvement over the previous $\Oh(n)$ time. Since the number of tripes $(c_u, h, k)$ is $\Oh(n^2)$ in total, there are at most $\Oh(n^2\log n)$ quadruples $(c_u, h, k_{max}, d)$ that can serve as parameters for $g$ (note that $0 \le d \le \log_2 n$). Every such quadruple is processed in constant time according to~(\ref{eq:main-recursion2}). Therefore, the overall running time of such algorithm with the use of memoization is $\Oh(n^2 \log n)$.

\begin{theorem}
An optimal skeleton tree for positive weights $w_1, w_2, \ldots, w_n$ can be constructed in $\Oh(n^2 \log n)$ time.\label{thm:square-algorithm}
\end{theorem}

\bibliographystyle{splncs04}
\bibliography{csr-34}
\end{document}